\documentclass[letterpaper, 10 pt,journal]{ieeetran}
\usepackage{graphics} % for pdf, bitmapped graphics files
\usepackage{epsfig} % for postscript graphics files
\usepackage{amsmath} % assumes amsmath package installed
\usepackage{amssymb}  % assumes amsmath package installed

 \IEEEpubidadjcol 
\IEEEoverridecommandlockouts

\usepackage{amsthm}

\newtheorem{assumption}{\bf Assumption}
\newtheorem{definition}{\bf Definition}

\newtheorem{proposition}{\bf Proposition}
\newtheorem{lemma}{\bf Lemma}

\newcommand{\prob}[1]{\mathrm{Prob}\left[#1\right]}
\newcommand{\expect}[1]{\mathbb{E}\left[#1\right]} 
\newcommand{\tr}[1]{\mathrm{tr}\left(#1\right)}

\usepackage{setspace}

\usepackage{tikz}
 \usetikzlibrary{arrows,patterns,mindmap,backgrounds,shadows}
\usetikzlibrary{backgrounds}
\usetikzlibrary{shapes,arrows,chains}

\usepackage{enumitem}
 
\usepackage{cite}
 \usepackage{hyperref}

\title{Recursively feasible stochastic predictive control using an interpolating initial state constraint \\ - extended version}
\author{Johannes K\"ohler, Melanie N. Zeilinger%
\thanks{Institute for Dynamic Systems and Control, ETH Zürich, Zürich CH-8092, Switzerland (e-mail:
[jkoehle$|$mzeilinger]@ethz.ch).}
}
\begin{document}
\IEEEoverridecommandlockouts
\IEEEpubid{\begin{minipage}{\textwidth}\ \\[12pt] \\ \\
         \copyright 2022 IEEE.  Personal use of this material is  permitted.  Permission from IEEE must be obtained for all other uses, in  any current or future media, including reprinting/republishing this material for advertising or promotional purposes, creating new  collective works, for resale or redistribution to servers or lists, or  reuse of any copyrighted component of this work in other works.
     \end{minipage}}
\maketitle

%
%%%%%%%%%%%%%%%%%%%%%%%%%%%%%%%%%%%%%%%%%%%%%%%%%%%%%%%%%%%%%%%%%%%%%%%%%%%%%%%%
\begin{abstract} 
We present a stochastic model predictive control (SMPC) framework for linear systems subject to possibly unbounded disturbances. 
State of the art SMPC approaches with closed-loop chance constraint satisfaction recursively initialize the nominal state based on the previously predicted nominal state or possibly the measured state under some case distinction.
We improve these initialization strategies by allowing for a continuous optimization over the nominal initial state in an interpolation of these two extremes. 
The resulting SMPC scheme can be implemented as one standard quadratic program and is more flexible compared to state-of-the-art initialization strategies.  
As the main technical contribution,  we show that the proposed SMPC framework also ensures closed-loop satisfaction of chance constraints and suitable performance bounds. 
\end{abstract} 
%\begin{IEEEkeywords} Predictive control for linear systems, 
% Stochastic optimal control \end{IEEEkeywords}
%!TEX root = ./SMPC_reach.tex
%%%%%%%%%%%%%%%%%%%%%%%%%%%%%%%%%%%%%%%%%%%%%%%%%%%%%%%%%%%%%%%%%%%%%%%%%%%%%%%
\section{Introduction}
Many challenging control problems are characterized by the need for performance and safety guarantees in the presence of model uncertainty. 
Model predictive control (MPC) is an optimization-based control method that accounts for general constraints and performance criteria~\cite{kouvaritakis2016model}. 
While there exist well-established methods to robustly account for model uncertainties in MPC, stochastic MPC (SMPC) approaches can take probabilistic information into account to reduce conservatism by allowing for a small user-defined probability of constraint violation~\cite{kouvaritakis2016model,farina2016stochastic,mesbah2016stochastic}. 
In this paper, we focus on the treatment of the initial state constraint in a receding horizon SMPC implementation and the resulting closed-loop properties in the form of performance and chance constraint satisfaction. 

\subsubsection*{Related work}
The reformulation of stochastic optimal control problems in terms of deterministic surrogates has been extensively studied in the literature considering: randomized methods using scenarios~\cite{schildbach2014scenario,lorenzen2016constraint,hewing2019scenario}, analytical reformulations~\cite{hewing2018stochastic,hewing2020recursively,paulson2020stochastic,farina2016stochastic,oldewurtel2008tractable,magni2009stochastic}, disturbance feedback optimization~\cite{oldewurtel2008tractable,magni2009stochastic,korda2014stochastic,yan2020stochastic,paulson2020stochastic}, or polynomial chaos expansion~\cite{mesbah2016stochastic,muhlpfordt2017comments}. 
The issue of recursive feasibility in SMPC formulations was early recognized~\cite{primbs2009stochastic}. 
Simple modifications include the relaxation of the corresponding chance constraints using penalties~\cite{oldewurtel2008tractable,paulson2020stochastic}, minimizing the probability of constraint violation~\cite{brudigam2021minimization}, (iteratively) adjusting the probability level in the chance constraints~\cite{magni2009stochastic,wang2021recursive}, or considering weaker discounted/weighted average probabilistic constraints~\cite{korda2014stochastic,yan2020stochastic}. 

By explicitly enforcing \textit{robust} recursive feasibility, closed-loop properties (performance, chance constraint satisfaction) can be ensured for SMPC schemes with chance constraints~\cite{cannon2010stochastic,kouvaritakis2010explicit,lorenzen2016constraint}.   
However, these approaches are not applicable to unbounded disturbances (e.g. Gaussian). 

In~\cite{farina2013probabilistic}, it was suggested to use a \textit{case distinction} to set the initial condition to the new measured state or the previous uncertain prediction. 
The same initialization strategy is also considered in various SMPC extensions~\cite{farina2015approach,farina2016model,li2021distributionally} and comparable re-set strategies are also used for stochastic reference governors~\cite{li2021chance}. 
In~\cite{hewing2018stochastic}, it was shown that under additional symmetry and unimodality conditions, a modified version of this SMPC formulation also satisfies the desired chance constraints in closed loop.
%\footnote{%
%As discussed in \cite[Remark 1]{farina2016model}, \cite[Remark 7]{li2021distributionally}, or shown in~\cite[Sec.~V]{hewing2018stochastic}, the case distinction SMPC \cite{farina2013probabilistic,farina2015approach,farina2016model,li2021distributionally} does in general \textit{not} result in closed-loop chance constraint satisfaction, compare also the discussion in Section~\ref{sec:discussion_proof}.} 
This paradigm was also adopted in recent SMPC extensions, e.g., \cite{mark2020stochastic,mark2021stochastic}.

To ensure closed-loop chance constraint satisfaction under more general conditions, the \textit{indirect feedback} SMPC was suggested in~\cite{hewing2020recursively}, where the feasible set is independent of the realized disturbances.  
Further extensions of this SMPC paradigm can be found in~\cite{hewing2019scenario,mark2021data} and a comparison to the \textit{case distinction} SMPC is given in~\cite{hewing2020performance}. 
Recently, a related initial state optimization has also been independently proposed in~\cite{schluter2022stochastic}.  
\subsubsection*{Contribution}
\IEEEpubidadjcol
In this paper, we improve the treatment of the initial state constraint in state-of-the-art SMPC formulations, primarily~\cite{hewing2018stochastic} and also \cite{hewing2020recursively}. 
We study linear systems with possibly unbounded disturbances subject to chance constraints on states and inputs. 
We relax the binary initial state constraint considered in~\cite{hewing2018stochastic} by allowing for a continuous interpolation of the nominal initial state between the new measured state and previous nominal prediction. 
The resulting SMPC scheme only requires the solution to one quadratic program (QP). 
As the main technical contribution, we prove that the proposed SMPC scheme ensures closed-loop satisfaction of the chance constraints, thus extending the results in~\cite{hewing2018stochastic}\footnote{%
The result applies if the tightened constraints are constructed using convex \textit{symmetric} probabilistic reachable sets and the distributions are (centrally convex) unimodal (e.g. Gaussian), which corresponds to the conditions in~\cite{hewing2018stochastic}.}.
Additionally, we significantly improve the performance guarantees in~\cite[Thm.~2]{hewing2018stochastic} by adopting the more direct cost function proposed in~\cite{hewing2020recursively}. 
Furthermore, we provide a comparison between the proposed SMPC formulation, the case distinction approach~\cite{hewing2018stochastic} and the indirect feedback formulation~\cite{hewing2020recursively} using an illustrative example.
In particular, these two SMPC formulations can be viewed as special cases of the considered interpolating initial state constraint, which explains why the proposed SMPC formulation is more flexible.  
\subsubsection*{Outline}
First, the problem setup and the proposed SMPC formulation are presented (Sec.~\ref{sec:problem}) and the closed-loop properties, i.e., recursive feasibility, chance constraint satisfaction, and performance, are shown (Sec.~\ref{sec:theory}).
Then, we compare the proposed approach with the SMPC approaches in~\cite{hewing2018stochastic} and \cite{hewing2020recursively} (Sec.~\ref{sec:discussion}) and conclude the paper (Sec.~\ref{sec:conclusion}).
This paper is an extended version of the paper~\cite{kohler2022recursively}, containing a detailed proof and an additional example in Appendix~\ref{app:performance} and \ref{app:example_symmetric}, respectively.

\subsubsection*{Notation}
The set of integers in the interval $[a,b]\subseteq\mathbb{R}$ is denoted by $\mathbb{I}_{[a,b]}$. 
The probability and the conditional probability are denoted by $\prob{A}$, $\prob{A|B}$. 
The expected value of a random variable $x$ is denoted by $\expect{x}$. 
We denote a variable $w$ with Gaussian distribution with mean $\mu$ and variance $\Sigma$ by $w\sim\mathcal{N}(\mu,\Sigma)$. 
The trace of a square matrix $A$ is denoted by $\tr{A}$. 
The Minkowski sum and the Pontryagin difference of two sets $\mathcal{S},\mathcal{T}$ are given by $\mathcal{S}\oplus\mathcal{T}=\{s+t|s\in\mathcal{S},~t\in\mathcal{T}\}$ and $\mathcal{S}\ominus\mathcal{T}=\{x|x+t\in\mathcal{S},\forall t\in\mathcal{T}\}$, respectively.

%!TEX root = ./SMPC_reach.tex
%%%%%%%%%%%%%%%%%%%%%%%%%%%%%%%%%%%%%%%%%%%%%%%%%%%%%%%%%%%%%%%%%%%%%%%%%%%%%%% 
\section{Problem setup and proposed SMPC formulation}
\label{sec:problem}
We first present the problem setup, the proposed SMPC formulation, and recall some properties of unimodal distributions. 
%!TEX root = ./SMPC_reach.tex
%%%%%%%%%%%%%%%%%%%%%%%%%%%%%%%%%%%%%%%%%%%%%%%%%%%%%%%%%%%%%%%%%%%%%%%%%%%%%%% 
\subsection{Setup}
We consider a linear time-invariant (LTI) system
\begin{align}
\label{eq:system}
x(k+1)=Ax(k)+Bu(k)+w(k),\quad  x(0)=x_0,
\end{align}
with state $x(k)\in\mathbb{R}^n$, input $u(k)\in\mathbb{R}^m$, time $k\in\mathbb{I}_{\geq 0}$, and disturbances $w(k)\in\mathbb{R}^n$.
The disturbances are assumed to be independent and identically distributed (i.i.d.) with distribution $w\sim\mathcal{Q}_{\mathrm{w}}$, zero mean, and variance $\Sigma_{\mathrm{w}}\succ 0$. 
We assume that $(A,B)$  is stabilizable. 
Hence, we can choose a feedback $K\in\mathbb{R}^{m\times n}$, such that $A_K:=A+BK$ is Schur. 
We consider chance constraints on the state and input of the form 
\begin{align}
\label{eq:chance_constraints}
\prob{(x(k),u(k))\in\mathcal{Z}}\geq p,\quad k\in\mathbb{I}_{\geq 0},
\end{align}
with a polytope $\mathcal{Z}\subseteq\mathbb{R}^{n+m}$ and some probability level $p\in(0,1)$.   
Furthermore, we consider a  linear quadratic stage cost  
\begin{align}
\label{eq:ell}
\ell(x,u) =x^\top Q x+x^\top q+u^\top R u+u^\top r,
\end{align}
 with matrices $Q,R\succeq 0$. 
The control goal is the minimization of the closed-loop cost, while satisfying the chance constraints~\eqref{eq:chance_constraints}. 
Chance constraints~\eqref{eq:chance_constraints} relax the deterministic constraint satisfaction requirement from (robust) MPC~\cite{kouvaritakis2016model} and thus allow for more flexible operation. 
This relaxation also becomes necessary in case of possibly unbounded disturbances, e.g., Gaussian, compare~\cite{farina2016stochastic,mesbah2016stochastic} for a general introduction to SMPC. 
Additional hard input constraints  $u_k\in\mathcal{U}$, $k\in\mathbb{I}_{\geq 0}$ can be enforced by considering open-loop stable systems with $K=0$ or by assuming bounded disturbances $w_k\in\mathcal{W}$, $k\in\mathbb{I}_{\geq 0}$. 

%!TEX root = ./SMPC_reach.tex
%%%%%%%%%%%%%%%%%%%%%%%%%%%%%%%%%%%%%%%%%%%%%%%%%%%%%%%%%%%%%%%%%%%%%%%%%%%%%%% 
\subsection{Proposed SMPC formulation}
At time $k\in\mathbb{I}_{\geq 0}$, we consider the following SMPC problem:  
\begin{subequations}
\label{eq:MPC}
\begin{align}
\label{eq:MPC_cost}
\min_{v_{\cdot|k},z_{\cdot|k},\lambda_k}&\mathcal{J}_N(x(k),z_{\cdot|k},v_{\cdot|k},\lambda_k)\\
\label{eq:MPC_interpolation}
\text{subject to }& z_{0|k}=(1-\lambda_k)z_{1|k-1}^\star+\lambda_k x(k),~ \lambda_k\in[0,1],\\
\label{eq:MPC_nominal_dynamics}
&z_{i+1|k}=Az_{i|k}+Bv_{i|k},~ i\in\mathbb{I}_{[0,N-1]},\\
\label{eq:MPC_nominal_tightened}
&(z_{i|k},v_{i|k})\in\overline{\mathcal{Z}}_{k+i},\quad i\in\mathbb{I}_{[0,N-1]},\\
&z_{N|k}\in\mathcal{X}_{\mathrm{f}},
\end{align}
\end{subequations}
with some later specified open-loop cost $\mathcal{J}_N$ (Sec.~\ref{sec:theory_performance}), a (polytopic) terminal set constraint $\mathcal{X}_{\mathrm{f}}$ (Sec.~\ref{sec:theory_rec}), and tightened constraints $\overline{\mathcal{Z}}_k\subseteq\mathcal{Z}$, $k\in\mathbb{I}_{\geq 0}$ (Sec.~\ref{sec:theory_constraints}).  
The optimization variables correspond to a nominal predicted state and input sequence $z_{\cdot|k}$, $v_{\cdot|k}$, and an interpolating variable $\lambda_k\in[0,1]$ for the initial nominal state. 
The minimum is denoted by $\mathcal{J}_N^\star(x(k),z_{1|k-1}^\star)$ and a minimizer is denoted by $v^\star_{\cdot|k}$, $z^\star_{\cdot|k}$, $\lambda^\star_k$.
Note that Problem~\eqref{eq:MPC} at time $k\in\mathbb{I}_{\geq 0}$ also depends on the nominal state $z_{1|k-1}^\star$ predicted at time $k-1$, which is initialized with $z_{1|-1}^\star=x_0$, similar to the SMPC approaches in~\cite{farina2013probabilistic,farina2015approach,farina2016stochastic,hewing2018stochastic,hewing2020recursively}. 
In closed loop, we apply the nominal input with the tube controller $K$, i.e.,
\begin{align}
\label{eq:u}
u(k)=v^\star_{0|k}+K(x(k)-z_{0|k}^\star),\quad k\in\mathbb{I}_{\geq 0}.
\end{align}
Feedback with respect to the measured state $x(k)$ is introduced in~\eqref{eq:u} and in the initial state constraint~\eqref{eq:MPC_interpolation}, which is comparable to standard tube-based MPC approaches (both robust and stochastic~\cite{kouvaritakis2016model,farina2016stochastic,hewing2018stochastic,hewing2020recursively}).
Since the cost $\mathcal{J}_N$ is based on the convex quadratic stage cost $\ell$~\eqref{eq:ell} (Sec.~\ref{sec:theory_performance}) and the tightened constraints $\overline{\mathcal{Z}}_k$ are polytopes (Sec.~\ref{sec:theory_constraints}), Problem~\eqref{eq:MPC} is a QP that can be efficiently solved. 
The SMPC formulation uses a nominal prediction~\eqref{eq:MPC_nominal_dynamics} for the constraints~\eqref{eq:MPC_nominal_tightened} and the true measured state $x(k)$ affects the initial condition~\eqref{eq:MPC_interpolation} and the cost function~\eqref{eq:MPC_cost}. 
The variable $\lambda_k\in[0,1]$ in~\eqref{eq:MPC_interpolation} allows for a continuous interpolation between the two initialization strategies considered in~\cite{hewing2018stochastic}.  
In Section~\ref{sec:discussion}, we discuss in more detail how the proposed approach is related to the SMPC approaches based on \textit{case distinction}~\cite{farina2013probabilistic,farina2015approach,hewing2018stochastic} and \textit{indirect feedback}~\cite{hewing2020recursively} and provide an illustrative example.   

%!TEX root = ./SMPC_reach.tex
%%%%%%%%%%%%%%%%%%%%%%%%%%%%%%%%%%%%%%%%%%%%%%%%%%%%%%%%%%%%%%%%%%%%%%%%%%%%%%% 
\subsection{Unimodality}
\label{sec:theory_unimodal}
In the following, we recap some properties of unimodal distributions, which are used for the theoretical analysis.  
\begin{definition}
\label{def:monotone}
\cite[Def.~3.2]{dharmadhikari1976multivariate} 
A distribution $\mathcal{Q}_{\mathrm{w}}$ in $\mathbb{R}^n$ is called monotone unimodal if for every symmetric convex set $\mathcal{R}\subseteq\mathbb{R}^n$ and every $x\in\mathbb{R}^n$, the quantity $\prob{w+cx\in\mathcal{R}}$ is non-increasing in $c\in[0,\infty]$. 
\end{definition}
\begin{lemma}
\label{lemma:unimodal_monotone} (adapted from \cite[Lemma~1]{hewing2018stochastic})
Suppose $x,w\in\mathbb{R}^n$ are independent random variables, $\mathcal{R}\subseteq\mathbb{R}^n$ is a convex symmetric set, and the distribution of $w$ is monotone unimodal. 
Then,  $\prob{cx+w\in\mathcal{R}}\geq \prob{\tilde{c}x+w\in\mathcal{R}}$, $\forall \tilde{c}\geq c\geq 0$. 
\end{lemma} 
\begin{definition}
\label{def:convex}
\cite[Def.~3.1]{dharmadhikari1976multivariate} 
A distribution $\mathcal{Q}_{\mathrm{w}}$ in $\mathbb{R}^n$ is called central convex unimodal if it is in the closed convex hull of the set of all uniform distributions on symmetric compact bodies in $\mathbb{R}^n$. 
\end{definition} 
\begin{proposition}
\label{prop:property}
\cite{dharmadhikari1976multivariate} 
Any central convex unimodal distribution is also monotone unimodal.
 Central convex unimodal distributions are closed under linear transformation and convolution. 
\end{proposition}  
 \begin{assumption}
\label{ass:disturbances_2} 
The distribution $\mathcal{Q}_{\mathrm{w}}$ is centrally convex unimodal. 
\end{assumption}
This condition includes many standard distributions, such as multivariate Gaussians, uniform distributions (over convex symmetric sets), or the Laplace distribution.  
Based on this restriction (Ass.~\ref{ass:disturbances_2}), Proposition~\ref{prop:property} ensures that the predicted error (i.e., a linear combination of the disturbances) is centrally convex unimodal and also monotone unimodal, which implies that Lemma~\ref{lemma:unimodal_monotone} can be used.  
Similar arguments were used in~\cite[Thm.~3]{hewing2018stochastic} to show closed-loop properties under a case distinction and in the next section we extend these arguments to the proposed interpolating initial state constraint.

%!TEX root = ./SMPC_reach.tex
%%%%%%%%%%%%%%%%%%%%%%%%%%%%%%%%%%%%%%%%%%%%%%%%%%%%%%%%%%%%%%%%%%%%%%%%%%%%%%% 
\section{Closed-loop analysis}
\label{sec:theory} 
In the following, we show that the proposed SMPC formulation provides all the desired closed-loop properties under the same conditions used in~\cite{hewing2018stochastic}. 
First,  we show recursive feasibility (Sec.~\ref{sec:theory_rec}).  
Then, as the main technical contribution, we investigate the error dynamics under the interpolating initial state constraint~\eqref{eq:MPC_interpolation} and show a suitable nestedness property (Sec.~\ref{sec:theory_error}). 
This property then allows us to prove closed-loop satisfaction of the chance constraints by using suitably tightened constraints $\overline{\mathcal{Z}}_{k}$ (Sec.~\ref{sec:theory_constraints}). 
Finally, we derive closed-loop performance bounds (Sec.~\ref{sec:theory_performance}).
%!TEX root = ./SMPC_reach.tex
%%%%%%%%%%%%%%%%%%%%%%%%%%%%%%%%%%%%%%%%%%%%%%%%%%%%%%%%%%%%%%%%%%%%%%%%%%%%%%% 
\subsection{Recursive feasibility}
\label{sec:theory_rec} 
Analogously to existing results in SMPC (cf.~\cite{farina2013probabilistic,hewing2018stochastic,hewing2020recursively}), recursive feasibility can be established by choosing a suitable (nominal) terminal set constraint $\mathcal{X}_{\mathrm{f}}$. 
\begin{assumption}
\label{ass:terminal_set} 
There exists a matrix $K_{\mathrm{f}}\in\mathbb{R}^{m\times n}$, such that for all $z\in\mathcal{X}_{\mathrm{f}}$, $(z,K_{\mathrm{f}}z)\in\overline{\mathcal{Z}}_k$, $k\in\mathbb{I}_{\geq 0}$ and $(A+BK_{\mathrm{f}})z\in\mathcal{X}_{\mathrm{f}}$. 
\end{assumption}
\begin{proposition}
\label{prop:recursive_feasible}
Let Assumption~\ref{ass:terminal_set} hold and suppose Problem~\eqref{eq:MPC} is feasible at $k=0$.
Then, Problem~\eqref{eq:MPC} is feasible for all $k\in\mathbb{I}_{\geq 0}$  for the resulting closed-loop system~\eqref{eq:system}, \eqref{eq:u}.
\end{proposition}
\begin{proof}
Given a feasible solution to Problem~\eqref{eq:MPC} at time $k\in\mathbb{I}_{\geq 0}$, a feasible candidate solution at time $k+1$ is given by  $\lambda_{k+1}=0$,
$v_{i|k+1}=v_{i+1|k}^\star$, $i\in\mathbb{I}_{[0,N-2]}$, $v_{N-1|k+1}=K_{\mathrm{f}}z_{N-1|k+1}$ using  Assumption~\ref{ass:terminal_set}. 
This candidate solution is independent of the disturbance realization $w$ and recursive feasibility follows with standard arguments~\cite[Thm.~1]{hewing2020recursively}. 
\end{proof} 
%!TEX root = ./SMPC_reach.tex
%%%%%%%%%%%%%%%%%%%%%%%%%%%%%%%%%%%%%%%%%%%%%%%%%%%%%%%%%%%%%%%%%%%%%%%%%%%%%%% 
\subsection{Predicted and closed-loop error}
\label{sec:theory_error}
In order to prove closed-loop satisfaction of the chance constraints~\eqref{eq:chance_constraints}, we require suitable properties for the closed-loop error. 
To this end, we show that for any convex symmetric set, the containment probability of the \textit{closed-loop} error is not lower than the corresponding probability of the initially predicted error (Prop.~\ref{prop:error_closed_loop}). 
The error between the measured state and the nominal (online optimized) state is given by
\begin{align}
\label{eq:error}
e(k):=&x(k)-z^\star_{0|k}\stackrel{\eqref{eq:MPC_interpolation}}{=}(1-\lambda_k^\star)(x(k)-z_{1|k-1}^\star).
\end{align}
We define the (uncertain) prediction for the state $x$  resulting from application of the stabilizing feedback~\eqref{eq:u} as
\begin{align}
\label{eq:state_prediction}
x_{i+1|k}
=&Ax_{i|k}+B(v_{i|k}^\star+K(x_{i|k}-z^\star_{i|k}))+w(k+i),
\end{align}
with $x_{0|k}=x(k)$ and $i\in\mathbb{I}_{[0,N-1]}$, $k\in\mathbb{I}_{\geq 0}$. 
Correspondingly, we also define the (uncertain) predicted error $e_{i|k}:=x_{i|k}-z_{i|k}^\star$, 
which satisfies the following recursion
\begin{align}
\label{eq:error_prediction_recursion}
e_{i+1|k}=&A_Ke_{i|k}+w(i+k), \quad e_{0|k}=e(k).
\end{align}
\begin{proposition}
\label{prop:error_closed_loop}
Let Assumption~\ref{ass:disturbances_2} hold and consider any convex symmetric set $\mathcal{R}\subseteq\mathbb{R}^n$. 
For system~\eqref{eq:system} under the control law~\eqref{eq:u} resulting from~\eqref{eq:MPC}, the error~\eqref{eq:error}--\eqref{eq:error_prediction_recursion}
satisfies
\begin{align}
\label{eq:set_inclusion}
\prob{e(k)\in\mathcal{R}}\geq \prob{e_{k|0}\in\mathcal{R}},\quad k\in\mathbb{I}_{\geq 0}.
\end{align}
\end{proposition}
\begin{proof}
The proof is an extension of~\cite[Thm.~3]{hewing2018stochastic}.  
At time $k-i$, the error is given by
\begin{align}
\label{eq:error_step}
&e(k-i)\stackrel{\eqref{eq:error}}{=}(1-\lambda_{k-i}^\star)(x(k-i)-z_{1|k-i-1}^\star). 
\end{align}
The predicted error is given by 
\begin{align}
\label{eq:error_prediction_lambda}
&e_{i|k-i}\stackrel{\eqref{eq:error_prediction_recursion}}{=}A_K^ie(k-i)+\underbrace{\sum_{j=0}^{i-1}A_K^{i-j-1}w(k-i+j)}_{=:\tilde{e}_{i|k-i}}\\
\label{eq:error_prediction_lambda_2}
\stackrel{\eqref{eq:error_step}}{=}&(1-\lambda_{k-i}^\star)\underbrace{A_K^i(x(k-i)-z_{1|k-i-1}^\star)}_{=:\hat{e}_{i|k-i}}+\tilde{e}_{i|k-i}.
\end{align}
Note that the term $\tilde{e}_{i|k-i}$ only depends on future disturbances (given time $k-i$), while the term $\hat{e}_{i|k-i}$ depends on terms known at time $k-i$. 
Given that $w$ is i.i.d., this implies that $\tilde{e}_{i|k-i}$ and $\hat{e}_{i|k-i}$ are independent.  
Furthermore, using  Assumption~\ref{ass:disturbances_2} and Proposition~\ref{prop:property}, $\tilde{e}_{i|k-i}$ is central convex unimodal and monotone unimodal (Def.~\ref{def:monotone}--\ref{def:convex}). 
Correspondingly, we can invoke Lemma~\ref{lemma:unimodal_monotone} with $c=1-\lambda_{k-i}^\star\in[0,1]$ and $\tilde{c}=1$:
\begin{align}
\label{eq:induction claim}
&\prob{e_{i|k-i}\in\mathcal{R}}\nonumber\\
\stackrel{\eqref{eq:error_prediction_lambda_2}}{=}&\prob{(1-\lambda_{k-i}^\star)\hat{e}_{i|k-i}+\tilde{e}_{i|k-i}\in\mathcal{R}}\nonumber\\
\stackrel{\mathrm{Lemma~\ref{lemma:unimodal_monotone}}}{\geq}& \prob{\hat{e}_{i|k-i}+\tilde{e}_{i|k-i}\in\mathcal{R}}\nonumber\\
=& \prob{e_{i+1|k-i-1}\in\mathcal{R}},
\end{align}
where the last equality holds with $\hat{e}_{i|k-i}+\tilde{e}_{i|k-i}=e_{i+1|k-i-1}$. 
Utilizing~\eqref{eq:induction claim} recursively yields~\eqref{eq:set_inclusion}. 
\end{proof}
This result ensures that probabilistic guarantees derived from the initially predicted error remain valid for the true closed-loop error if: the corresponding set $\mathcal{R}$ is convex symmetric  and the distribution is centrally convex unimodal. 
This result is a stochastic counterpart to the nestedness properties usually invoked in robust MPC (cf., e.g.~\cite[Fig.~7.3]{kouvaritakis2016model}). 
 
%!TEX root = ./SMPC_reach.tex
%%%%%%%%%%%%%%%%%%%%%%%%%%%%%%%%%%%%%%%%%%%%%%%%%%%%%%%%%%%%%%%%%%%%%%%%%%%%%%% 
\subsection{Tightened constraints and chance constraint satisfaction}
\label{sec:theory_constraints} 
A set $\mathcal{R}$ satisfying $\prob{e_{k|0}\in\mathcal{R}}\geq p$ is called a $k$-step \textit{probabilistic reachable set} (PRS)~\cite[Def.~4]{hewing2018stochastic}, \cite[Def.~2]{hewing2020recursively}, where the predicted error $e_{k|0}$ satisfies the linear system~\eqref{eq:error_prediction_recursion} with $e_{0|0}=0$. 
A simple PRS is given by $\mathcal{R}_k=\{x|x^\top \Sigma_{\mathrm{x},k}^{-1}x\leq \tilde{p}\}$, 
\begin{align}
\label{eq:variance_prediction}
\Sigma_{\mathrm{x},k+1}=A_K \Sigma_{\mathrm{x},k} A_K^\top+\Sigma_{\mathrm{w}},~k\in\mathbb{I}_{\geq 0},\quad \Sigma_{\mathrm{x},0}=0,
\end{align}
with $\tilde{p}$ computed using the Chebyshev bounds or directly a Gaussian distribution, compare~\cite[Lemma 3, Rk.~3]{hewing2018stochastic}, \cite[Sec.~3.2]{hewing2020recursively}.  
In case only a subspace $y_{\mathrm{e}}=(C+DK)e\in\mathbb{R}^{n_{\mathrm{y}}}$, $n_{\mathrm{y}}<n$ of the error $e$ is relevant for the constraints $\mathcal{Z}$~\eqref{eq:chance_constraints},  the PRS should be be constructed by directly using the distribution of the lower dimensional output variable $y_{\mathrm{e}}$~\cite[Equ.~(18)]{li2021chance}, \cite[Sec.~3.3]{hewing2020recursively}. 
In case of non-Gaussian disturbances, the Chebyshev bounds can be very conservative, which can be mitigated by using scenario-based approaches~\cite{lorenzen2016constraint,hewing2019scenario} or polynomial
chaos expansion~\cite[Sec.~3.A]{muhlpfordt2017comments} to compute the PRS.  
The presented theoretical results are applicable to any RPS, assuming the sets are chosen to be symmetric and convex. 
\begin{assumption}
\label{ass:PRS}
The sets $\mathcal{R}_k$ are chosen to be symmetric, convex, and satisfy $\prob{e_{k|0}\in\mathcal{R}_k}\geq p$, $k\in\mathbb{I}_{\geq 0}\cup\{\infty\}$. 
\end{assumption}
The restriction to \textit{symmetric} PRS (Ass.~\ref{ass:PRS}) can introduce conservatism but is essential for the presented theoretical results,  compare the discussion in Section~\ref{sec:discussion} and the example in~Appendix~\ref{app:example_symmetric} for details. 
Given the PRS, we compute the tightened constraints using
\begin{align}
\label{eq:tightend_constraints}
\overline{\mathcal{Z}}_k:=\mathcal{Z}\ominus(\mathcal{R}_k\times K\mathcal{R}_k),\quad k\in\mathbb{I}_{\geq 0}.
\end{align}
Closed-loop chance constraint satisfaction holds by applying the main technical result (Prop.~\ref{prop:error_closed_loop}) and using the constraint tightening~\eqref{eq:tightend_constraints} with Assumption~\ref{ass:PRS}, analogously to~\cite[Cor.~1]{hewing2018stochastic}. 
\begin{proposition}
\label{prop:chance_constraint}
Let Assumptions~\ref{ass:disturbances_2}--\ref{ass:PRS} hold.  
Suppose that Problem~\eqref{eq:MPC} is feasible at $k=0$.
Then, the resulting closed-loop system~\eqref{eq:system}, \eqref{eq:u} satisfies the chance constraints~\eqref{eq:chance_constraints}. 
\end{proposition}
\begin{proof}
In case $e(k)\in\mathcal{R}_k$,  the constraint tightening implies
\begin{align*}
&(x(k),u(k))\stackrel{\eqref{eq:u},\eqref{eq:error}}{=}(z_{0|k}^\star,v_{0|k}^\star)+(e(k),Ke(k))\\ 
\stackrel{\eqref{eq:MPC_nominal_tightened}}{\in}&\overline{\mathcal{Z}}_k\oplus (\mathcal{R}_k\times K\mathcal{R}_k)\stackrel{\eqref{eq:tightend_constraints}}{\subseteq}\mathcal{Z},\quad k\in\mathbb{I}_{\geq 0}.
\end{align*}
Thus, chance constraint satisfaction~\eqref{eq:chance_constraints} follows from
\begin{align*}
&\prob{e(k)\in\mathcal{R}_k}
\stackrel{\eqref{eq:set_inclusion}}{\geq} \prob{e_{k|0}\in\mathcal{R}_k}\stackrel{\mathrm{Ass.}~\ref{ass:PRS}}{\geq} p,\quad k\in\mathbb{I}_{\geq 0}.&\qedhere
\end{align*}
\end{proof}
The (joint) chance constraint~\eqref{eq:chance_constraints} is often replaced by a set of individual chance constraints using, e.g. Boole’s inequality and risk allocation~\cite{paulson2020stochastic}. 
In this case, one should construct PRS (Ass.~\ref{ass:PRS}) for each individual chance constraint and tighten the constraints respectively (cf.~\cite[Equ.~5]{hewing2020recursively}). 

%!TEX root = ./SMPC_reach.tex
%%%%%%%%%%%%%%%%%%%%%%%%%%%%%%%%%%%%%%%%%%%%%%%%%%%%%%%%%%%%%%%%%%%%%%%%%%%%%%% 
\subsection{Performance analysis}
\label{sec:theory_performance}
We show that the closed-loop performance is on average no worse than applying the state feedback $u=Kx$, where $K$ is the feedback in~\eqref{eq:u}, which significantly improves the bound for the case-distinction SMPC~\cite{hewing2018stochastic} (cf. discussion in Section~\ref{sec:discussion}). 
\begin{assumption}
\label{ass:terminal_cost} 
Assumption~\ref{ass:terminal_set} holds with $K_{\mathrm{f}}=K$. 
The terminal cost $V_{\mathrm{f}}(x)=x^\top P_{\mathrm{f}}x+p_{\mathrm{f}}^\top x$ is chosen such that $V_{\mathrm{f}}((A+BK)x)= V_{\mathrm{f}}(x)-\ell(x,Kx)$ for all $x\in\mathbb{R}^n$. 
\end{assumption}
The restriction $K_{\mathrm{f}}=K$ is also needed in comparable performance results in SMPC schemes~\cite{lorenzen2016constraint,hewing2020recursively}. As in~\cite{hewing2020recursively}, the cost function in the SMPC~\eqref{eq:MPC} is chosen as the finite-horizon expected cost conditioned on the measured state $x(k)$:  
\begin{align*} 
\mathcal{J}_N(x(k),z_{\cdot|k},v_{\cdot|k},\lambda_k):=\expect{\sum_{i=0}^{N-1}\ell(x_{i|k},u_{i|k})+V_{\mathrm{f}}(x_{N|k})},
\end{align*}
where the predicted state and input are given by $x_{0|k}=x(k)$, \eqref{eq:state_prediction} and $u_{i|k}=v_{i|k}+K(x_{i|k}-z_{i|k})$. 
The corresponding mean is given by $\overline{x}_{i+1|k}=A\overline{x}_{i|k}+B\overline{u}_{i|k}$, $\overline{u}_{i|k}=v_{i|k}+K(\overline{x}_{i|k}-z_{i|k})$, $\overline{x}_{0|k}=x(k)$, which is \textit{not} equivalent to the nominal trajectory $z_{\cdot|k}$, $v_{\cdot|k}$ due to the interpolating initial state constraint~\eqref{eq:MPC_interpolation} in case $\lambda_k\neq 1$. 
The expected cost is equivalent to
\begin{align}
\label{eq:J_expected_2}
&\mathcal{J}_N(x(k),z_{\cdot|k},v_{\cdot|k},\lambda_k)
=\sum_{i=0}^{N-1}\ell(\overline{x}_{i|k},\overline{u}_{i|k})+V_{\mathrm{f}}(\overline{x}_{N|k})\\
&+\sum_{i=0}^{N-1}\tr{(Q+K^\top RK)\Sigma_{\mathrm{x},i}}+\tr{P_{\mathrm{f}}\Sigma_{\mathrm{x},N}},\nonumber
\end{align}
with the state variance  $\Sigma_{\mathrm{x},i}$ according to~\eqref{eq:variance_prediction}. 
\begin{proposition}
\label{prop:performance}
Let Assumptions~ \ref{ass:terminal_set} and \ref{ass:terminal_cost} hold. 
Suppose that Problem~\eqref{eq:MPC} is feasible at $k=0$. 
Then, for all $k\in\mathbb{I}_{\geq 0}$:
\begin{align}
\label{eq:value_expected_next}
&\expect{\mathcal{J}_N^\star(x(k+1),z_{1|k}^\star)-\mathcal{J}_N^\star(x(k),z_{1|k-1}^\star)}\\
\leq&-\ell(x(k),u(k))+ \tr{P_{\mathrm{f}}\Sigma_{\mathrm{w}}}.\nonumber
\end{align}
Furthermore, if the stage cost $\ell(x,u)$ is lower bounded, then 
\begin{align}
\label{eq:performance_bound}
\limsup_{K\rightarrow \infty}\dfrac{1}{K}\sum_{k=0}^{K-1}\expect{\ell(x(k),u(k))}\leq \tr{P_{\mathrm{f}}\Sigma_{\mathrm{w}}}.
\end{align}
\end{proposition}
\begin{proof}
The proof follows from~\cite[Thm.~3]{hewing2020recursively}, where the same cost $\mathcal{J}_N$ and candidate solution is considered, compare Appendix~\ref{app:performance} for a detailed exposition. 
\end{proof}
Inequality~\eqref{eq:performance_bound} ensures that the asymptotic average performance is no worse than the performance of the linear controller $u=Kx$. 
For computational reasons, one can add a penalty on $\lambda_k^2$ to the cost function,
which penalizes $\|z_{0|k}-z_{1|k-1}^\star\|$ and hence does not affect the performance guarantees.

%!TEX root = ./SMPC_reach.tex
%%%%%%%%%%%%%%%%%%%%%%%%%%%%%%%%%%%%%%%%%%%%%%%%%%%%%%%%%%%%%%%%%%%%%%%%%%%%%%% 
\section{Discussion and comparison}
\label{sec:discussion}
First, we clarify the relation of the proposed SMPC approach to the case distinction SMPC~\cite{hewing2018stochastic} and the indirect feedback SMPC~\cite{hewing2020recursively}. 
Then, we consider an illustrative example to demonstrate the performance benefits of the proposed SMPC formulation. 
Finally, we contrast the presented proof and resulting guarantees to related results in the SMPC literature. 
%!TEX root = ./SMPC_reach.tex
%%%%%%%%%%%%%%%%%%%%%%%%%%%%%%%%%%%%%%%%%%%%%%%%%%%%%%%%%%%%%%%%%%%%%%%%%%%%%%% 
\subsubsection{Comparison to case distinction SMPC}
By considering the special case $\lambda_k\in\{0,1\}$, the proposed SMPC formulation is comparable to the case distinction SMPC proposed in~\cite{hewing2018stochastic}.  
First, considering the implementation, the proposed SMPC approach is clearly preferable since one QP needs to be solved while the case distinction SMPC might need to solve two QPs.  
Second, the tightened constraints and the theoretical results regarding closed-loop chance constraint satisfaction (Prop.~\ref{prop:chance_constraint}) are equivalent to~\cite{hewing2018stochastic}. 
Hence, regarding the set of feasible decisions, the proposed SMPC approach is more flexible as the full interval  $\lambda_k\in[0,1]$ is feasible, and not only the extreme values $\lambda_k\in\{0,1\}$, compare also the numerical example later (Sec.~\ref{sec:example_LQR}). 
The tightened constraints in other case distinction SMPC schemes~\cite{farina2013probabilistic,farina2015approach,farina2016model,li2021distributionally} can be less conservative as no symmetry condition is used (Ass.~\ref{ass:PRS}) and additionally the time index $k$ in tightened constraints $\overline{\mathcal{Z}}_{k+i}$~\eqref{eq:MPC_nominal_tightened} is re-initialized to $0$ whenever $\lambda_k^\star=1$.  
However, as also discussed in~\cite{hewing2018stochastic} and shown in an additional example in Appendix~\ref{app:example_symmetric}, such approaches do \textit{not} result in closed-loop chance constraint satisfaction.

It is important to note that there exists a subtle difference regarding the cost function $\mathcal{J}_N$ considered in present paper and the one typically used in case distinction SMPC formulations~\cite{hewing2018stochastic,farina2013probabilistic}. 
In particular, in~\cite{farina2013probabilistic}, the expected cost is not conditioned on the known measured state $x(k)$, but for $\lambda_k\neq 1$ it is conditioned on the previously predicted state distribution. As a result, the derived performance bound deteriorates with an increasing prediction horizon $N$ (cf.~\cite[Thm.~1]{farina2013probabilistic}, \cite[Thm.~1]{farina2015approach}).
In~\cite{hewing2018stochastic} (cf. also~\cite[Sec.~III.D]{farina2013probabilistic}, \cite{mark2020stochastic,mark2021stochastic}) it is suggested to use only the nominal state and input in the cost and use $\lambda_k=1$ whenever feasible. 
The resulting performance bound utilizes Lipschitz continuity of the value function and also fails to recover the performance of the linear feedback $u=Kx$. 
These established performance bounds for  the case distinction SMPC scheme ($\lambda_k\in\{0,1\}$) can be directly improved by adopting 
the proposed cost (Sec.~\ref{sec:theory_performance}) and the corresponding theoretical analysis (Prop.~\ref{prop:performance}), which ensures a performance no worse than the tube feedback $u=Kx$. 
 
%!TEX root = ./SMPC_reach.tex
%%%%%%%%%%%%%%%%%%%%%%%%%%%%%%%%%%%%%%%%%%%%%%%%%%%%%%%%%%%%%%%%%%%%%%%%%%%%%%% 
\subsubsection{Comparison to indirect feedback SMPC}
The indirect feedback SMPC~\cite{hewing2020recursively} can be directly recovered as a special case of the proposed SMPC by setting $\lambda_k=0$. 
Hence, given the same constraint tightening, the proposed SMPC is more flexible. 
However, the theoretical guarantees in the indirect feedback SMPC are more general, and are applicable to correlated disturbances~\cite{hewing2020recursively}, nonlinear dynamics (cf. \cite{hewing2019scenario}), and non-symmetric PRS~\cite{hewing2020recursively}.
An additional example investigating the trade-off between symmetric PRS and the relaxed initial state constraint can be found in Appendix~\ref{app:example_symmetric}.

%!TEX root = ./SMPC_reach.tex
%%%%%%%%%%%%%%%%%%%%%%%%%%%%%%%%%%%%%%%%%%%%%%%%%%%%%%%%%%%%%%%%%%%%%%%%%%%%%%% 
\begin{table*}
\caption{Example: average performance and chance constraint satisfaction.\vspace{-1mm}}
\begin{tabular}{c|c|c|c|c|c|c|c|c|c}
Method&LQR&Proposed&Case ($\lambda_k\in\{0,1\}$)&Indirect \cite{hewing2020recursively}&Nominal&$u=Kx$&Case~\cite[Alg.~1]{hewing2018stochastic}\\
\hline
$\expect{\ell}$&$75.0\%$&$82.0\%$&$82.3\%$&$83.0\%$&$91.1\%$&$100.0\%$&$153.0\%$\\
$\prob{(x,u)\in\mathcal{Z}}$&$68.3\%$&$82.6\%$&$82.7\%$&$82.7\%$&$100.0\%$&$91.7\%$&$93.5\%$
\end{tabular}
\label{tab:academic_LQR}
\vspace{-1mm}
\end{table*}
\subsubsection{Numerical example} 
\label{sec:example_LQR}
The following example can be motivated by an inventory control problem, which does not require hard input constraints. 
We consider an integrator $x(k+1)=x(k)+u(k)+w(k)$, $w\in\mathcal{N}(0,1)$, $x_0=0$, with $\ell(x,u)=x^2$, i.e., in the absence of constraints the optimal feedback is the deadbeat controller $u=K_{\mathrm{LQR}}x=-x$. 
The tube controller is given by $K=-0.5$ and the PRS $\mathcal{R}_k$ are constructed using the stationary variance $\Sigma_\infty$~\cite[Rk.~3]{hewing2018stochastic}.  
We have a probabilistic input constraint $\prob{u\in[-1,1]}\geq p=80.61\%$, which is chosen such that $\overline{\mathcal{Z}}_k=\{(x,u)|~u\in[-0.25,0.25]\}$, $k\in\mathbb{I}_{\geq 0}$.   
This ensures that a more aggressive feedback than $u=Kx$ is feasible. 
The resulting average performance and probabilistic constraint satisfaction can be seen in Table~\ref{tab:academic_LQR}. We compare the proposed SMPC in~\eqref{eq:MPC} with the new interpolating initial state constraint to the state feedback $u=Kx$, $u=K_{\mathrm{LQR}}x$, the indirect feedback~\cite{hewing2020recursively}, the case distinction SMPC ($\lambda_k\in\{0,1\}$), the case distinction SMPC implementation suggested in~\cite[Alg.~1]{hewing2018stochastic}, and a nominal MPC\footnote{%
The implementation in~\cite[Alg.~1]{hewing2018stochastic} chooses $\lambda_k^\star=1$ whenever feasible, compare also~\cite[Sec.~III.D]{farina2013probabilistic}. 
For the nominal MPC, we implemented the proposed SMPC formulation with $K=0$, which results in nominal input constraints with a somewhat relaxed terminal set constraint.}.
The implementations use a simple terminal equality constraint $\mathcal{X}_{\mathrm{f}}=\{0\}$, a prediction horizon of $N=10$, and are simulated over $40$ steps with $10^4$ random disturbance realizations $w$. 
As expected, all the SMPC formulations satisfy the chance constraints~\eqref{eq:chance_constraints} in closed loop (Prop.~\ref{prop:chance_constraint}) and achieve a performance no worse than the tube feedback $u=Kx$ (Prop.~\ref{prop:performance}). 
One exception is the implementation suggested in~\cite[Agl. 1]{hewing2018stochastic}, which provides an unnecessarily large probability of constraint satisfaction and thus results in a significant performance deterioration, compare also~\cite{hewing2020performance} for a detailed discussion of this effect.
Furthermore, the proposed approach ($\lambda_k\in[0,1]$) outperforms the case distinction ($\lambda_k\in\{0,1\}$), which outperforms the indirect feedback ($\lambda_k=0$). 
This can be best explained by the implicit input constraint: $u(k)\in (1-\lambda_k) K(x(k)-z_{1|k-1}^\star)\oplus[-0.25,0.25]$.
There, $\lambda_k=0$~\cite{hewing2020recursively} is restrictive if the tube feedback $K(x(k)-z_{1|k}^\star)$ pushes the state $x$ in the wrong direction and $\lambda_k\in\{0,1\}$~\cite{hewing2018stochastic} may result in a disjoint set for $x(k)-z_{1|k-1}^\star$ large. 
 
%!TEX root = ./SMPC_reach.tex
%%%%%%%%%%%%%%%%%%%%%%%%%%%%%%%%%%%%%%%%%%%%%%%%%%%%%%%%%%%%%%%%%%%%%%%%%%%%%%% 
\subsubsection{A comment on simpler arguments in the  SMPC literature} 
\label{sec:discussion_proof}
At first glance, the derivation in Proposition~\ref{prop:chance_constraint} and similarly in~\cite{hewing2018stochastic} based on monotone unimodality and symmetry might seem unnecessarily complicated in contrast to many SMPC approaches based on case distinction~\cite{farina2013probabilistic,farina2015approach,farina2016model,li2021distributionally,li2021chance}. 
However, as discussed in~\cite[Remark 1]{farina2016model}, \cite[Remark 7]{li2021distributionally} and shown in the example in~\cite[Sec.~V]{hewing2018stochastic}) (cf. also the example in Appendix~\ref{app:example_symmetric}), other approaches that drop these requirements typically do \textit{not}
satisfy the chance constraints in \textit{closed-loop}. 
To understand this difference, we exemplary revisit the proof in~\cite[Prop. 2]{li2021chance} for a stochastic reference governor, where the implicit arguments in many SMPC papers are made more rigorously.\footnote{%
The authors in~\cite{li2021chance} are currently preparing a \textit{corrigendum}, that avoids these arguments in the proof by also utilizing tools similar to~\cite{hewing2018stochastic}.
} 
Suppose we wish to show~\eqref{eq:chance_constraints} for some fixed but arbitrary $k$ with a case distinction SMPC ($\lambda_k\in\{0,1\}$). 
For any realization of $w$, we can define a time $\tau=\max\{k'\in\mathbb{I}_{[0,k]}|~\lambda_{k'}^\star=1\}$, i.e., the most recent time where the SMPC was re-initialized with the new measured state. 
Then, we write the chance constraints using conditional probabilities (cf.~\cite[Equ.~(43)]{li2021chance}):
\begin{align}
\label{eq:alternative_proof}
&\prob{(x(k),u(k))\in\mathcal{Z}}\\
=&\sum_{i=0}^k\prob{(x(k),u(k))\in\mathcal{Z}|\tau=i}\cdot \prob{\tau=i}.\nonumber
\end{align}
The initialization $\lambda_0^\star=1$ ensures $\sum_{i=0}^k\prob{\tau=i}=1$ and hence the desired bound~\eqref{eq:chance_constraints} follows directly from~\eqref{eq:alternative_proof} if $\prob{(x(k),u(k))\in\mathcal{Z}|\tau=i}\geq p$. 
By construction, feasibility of any properly designed SMPC at some time $i$ should imply chance constraint satisfaction for all future time (until the next re-conditioning).
However, this is not the same as conditioning on $\tau=i$.  
In particular, $\tau=i$ also implies that for all future times  $k'>i$, the SMPC with $x(k')=z_{0|k'}^\star$ is infeasible. 
Thus, we cannot use these simple arguments to show~\eqref{eq:chance_constraints} based on~\eqref{eq:alternative_proof}. 
With this discussion, we wish to highlight the need for the more involved arguments from Proposition~\ref{prop:error_closed_loop} and the symmetry/unimodality condition to ensure closed-loop chance constraint satisfaction.

\subsubsection{Summary}
The proposed SMPC is clearly superior to the \textit{case distinction} SMPC~\cite{hewing2018stochastic} due to the simpler implementation, less restrictive constraints, and better performance guarantees, while considering the exact same assumptions and design. 
Compared to the \textit{indirect feedback} SMPC~\cite{hewing2020recursively}, we achieved better performance due to the less restrictive constraints. 
However, the guarantees in~\cite{hewing2020recursively} are applicable under less restrictive conditions. 
In Appendix~\ref{app:example_symmetric}, we provide an additional example that investigates the conservatism of using symmetric PRS (Ass.~\ref{ass:PRS}). 
Therein, we empirically find that the additional degrees of freedom in the initial state constraints can compensate the conservatism of using symmetric PRS (Ass.~\ref{ass:PRS}). 
Thus, in the considered setting ($w$ i.i.d. \& Ass.~\ref{ass:disturbances_2} holds), the proposed SMPC  approach consistently outperformed the state-of-the-art SMPC formulations~\cite{hewing2018stochastic,hewing2020recursively}, with no clear drawback.
Lastly, we clarified that the restriction to symmetric PRS (Ass.~\ref{ass:PRS}) is indeed needed for \textit{closed-loop chance constraint satisfaction}, while some of the simpler existing proofs for case distinction SMPC (cf. \cite{farina2013probabilistic,farina2015approach,farina2016model,li2021distributionally,li2021chance}) might \textit{not} yield similar guarantees.   
%!TEX root = ./SMPC_reach.tex
%%%%%%%%%%%%%%%%%%%%%%%%%%%%%%%%%%%%%%%%%%%%%%%%%%%%%%%%%%%%%%%%%%%%%%%%%%%%%%%
\section{Conclusion}
\label{sec:conclusion}
We have presented an SMPC scheme for linear systems with possibly unbounded disturbances, providing closed-loop chance constraint satisfaction and suitable performance bounds.   
The proposed SMPC formulation specifically improves the case distinction in prior SMPC approaches~\cite{hewing2018stochastic} (and similarly in~\cite{farina2013probabilistic,farina2015approach}) addressing recursive feasibility by interpolating between two initial states.  
The resulting SMPC approach only requires the solution of one QP and is applicable to a rather general class of linear stochastic optimal control problems. 
We used a simple example to argue why the proposed SMPC formulation allows for more flexibility and hence improved performance compared to the case distinction SMPC~\cite{hewing2018stochastic} and the indirect feedback SMPC~\cite{hewing2020recursively}. 
We expect that many of the recent SMPC results that build on one of these two SMPC approaches for the initial state constraint (cf., e.g., \cite{farina2015approach,farina2016model,li2021distributionally,li2021chance,mark2020stochastic,mark2021stochastic,mark2021data}) can be improved/extended by adopting the proposed interpolation for the initial state.   
\subsubsection*{Acknowledgment}
We would like to thank our colleague K. Wabersich for helpful discussions on the initial idea.
\bibliographystyle{ieeetran}  
\bibliography{Literature} 
\appendix
%!TEX root = ./SMPC_reach.tex
%%%%%%%%%%%%%%%%%%%%%%%%%%%%%%%%%%%%%%%%%%%%%%%%%%%%%%%%%%%%%%%%%%%%%%%%%%%%%%% 
\subsection{Proof  - Proposition~\ref{prop:performance}}
\label{app:performance}
The result in Proposition~\ref{prop:performance} follows using the arguments in~\cite[Thm.~3]{hewing2020recursively} using  the same cost $\mathcal{J}_N$ and candidate solution.
The following proof is added to ensure a self-contained exposition. 
\begin{proof}
We first established the (expected) decrease condition on the value function~\eqref{eq:value_expected_next} and then show that this implies the performance bound
\eqref{eq:performance_bound}.\\
\textbf{Part I: }
We abbreviate the quadratic norm w.r.t. a positive semi-definite matrix $Q$ by $\|x\|_Q^2=x^\top Q x$.
Define $\overline{u}_{N|k}^\star:=K\overline{x}_{N|k}^\star$, $\overline{x}_{N|k+1}^\star:=(A+BK)\overline{x}_{N|k}^\star$. 
Using the feasible candidate solution from Proposition~\ref{prop:recursive_feasible}, the predicted mean is given by $\overline{x}_{i|k+1}=\overline{x}_{i+1|k}^\star+A_K^iw(k)$, $\overline{u}_{i|k+1}=\overline{u}_{i+1|k}^\star+KA_K^iw(k)$, $i\in\mathbb{I}_{[0,N-1]}$. 
Furthermore, since $K=K_{\mathrm{f}}$, we have $\overline{x}_{N|k+1}=\overline{x}_{N+1|k}^\star+A_K^Nw(k)$. 
Considering the cost~\eqref{eq:J_expected_2}, this implies
\begin{align*}
&\expect{\mathcal{J}_N^\star(x(k+1),z_{1|k}^\star)-\mathcal{J}_N^\star(x(k),z_{1|k-1}^\star)}\\
&+\ell(x(k),u(k))\\
\leq& \expect{\sum_{i=0}^{N-1}\ell(\overline{x}_{i|k+1},\overline{u}_{i|k+1})-\ell(\overline{x}_{i+1|k}^\star,\overline{u}^\star_{i+1|k})}\\
&+\expect{V_{\mathrm{f}}(\overline{x}_{N|k+1})\underbrace{-\ell(\overline{x}_{N|k}^\star,\overline{u}_{N|k}^\star)-V_{\mathrm{f}}(\overline{x}_{N|k}^\star)}_{\stackrel{\text{Ass.~}\ref{ass:terminal_cost}}{=}-V_{\mathrm{f}}(\overline{x}_{N+1|k}^\star)}}\\
=&\sum_{i=0}^{N-1}\expect{\|A_K^iw(k)\|_{Q+K^\top R K}^2}+\expect{\|A_K^Nw(k)\|_{P_{\mathrm{f}}}^2}\\
\stackrel{\text{Ass.~}\ref{ass:terminal_cost}}{=} &\expect{\|w(k)\|_{P_{\mathrm{f}}}^2}=\tr{P_{\mathrm{f}}\Sigma_{\mathrm{w}}},
\end{align*}
where the second to last equality used $A_K^\top P_{\mathrm{f}}A_K+Q+K^\top R K= P_{\mathrm{f}}$ from Assumption~\ref{ass:terminal_cost} in a telescopic sum. \\
%The proof follows standard arguments form stochastic MPC \cite[Cor.~2]{hewing2018stochastic}, \cite[Thm.~3]{hewing2020recursively}.\\
\textbf{Part II: } Given that $\ell$ is lower bounded and $N,\Sigma_{\mathrm{w}}$ are finite, $\mathcal{J}^\star_N$ admits a lower bound. 
Hence we can use~\eqref{eq:value_expected_next} in a telescopic sum to arrive at the average performance bound~\eqref{eq:performance_bound}, compare, e.g., \cite[Cor.~2]{hewing2018stochastic}.
\end{proof}

%!TEX root = ./SMPC_reach.tex
%%%%%%%%%%%%%%%%%%%%%%%%%%%%%%%%%%%%%%%%%%%%%%%%%%%%%%%%%%%%%%%%%%%%%%%%%%%%%%% 
\subsection{Example - conservatism of symmetric sets}
\label{app:example_symmetric}
It is well known that in case of individual half-space constraints,  less conservative bounds can be obtained by using (non-symmetric) half-space based PRS~\cite{lorenzen2016constraint},\cite[Sec.~3.2]{hewing2020recursively}.
In the following, we empirically investigate the conservatism of symmetric PRS (Ass.~\ref{ass:PRS}). 

We consider a scalar example $x(k+1)=0.75x(k)+u(k)+w(k)$, $w\in\mathcal{N}(0,1)$, with a probabilistic constraint $\prob{x\geq -2}\geq p=81.4\%$, no stabilizing feedback ($K=0$), and stage cost $\ell(x,u)=u$. 
The example is chosen such that the tightened constraint $\overline{\mathcal{Z}}_\infty=\{(x,u)|x\geq 0\}$ is always active. 
In this example, the case-distinction SMPC and the proposed SMPC are equivalent and hence we only compare to the indirect feedback SMPC. 
In particular, we consider both approaches using either symmetric or non-symmetric PRS.
The results can be seen in Figure~\ref{fig:academic}.
\begin{figure}[hbtp]
\begin{center}
\includegraphics[width=0.395\textwidth]{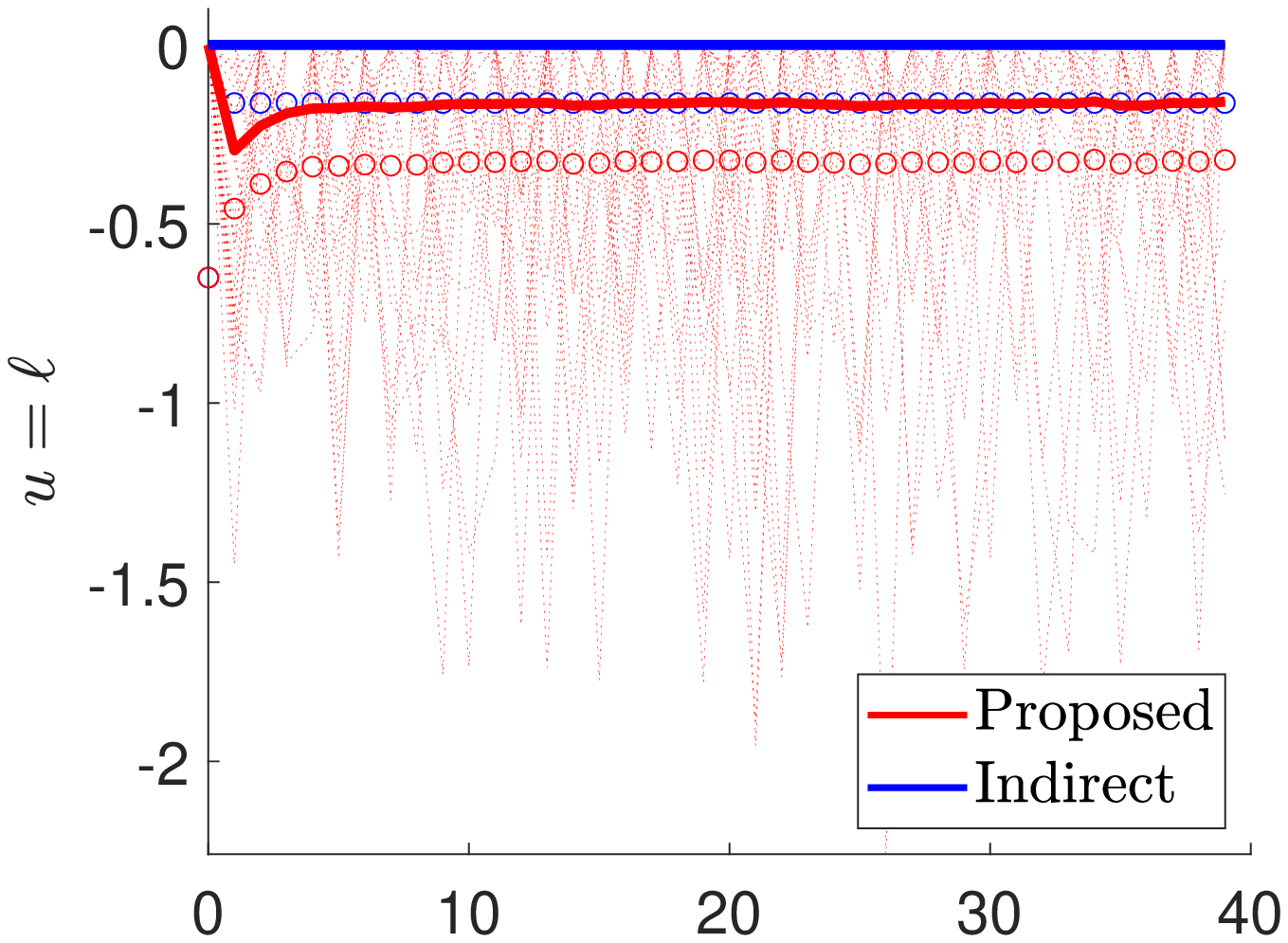}
\includegraphics[width=0.395\textwidth]{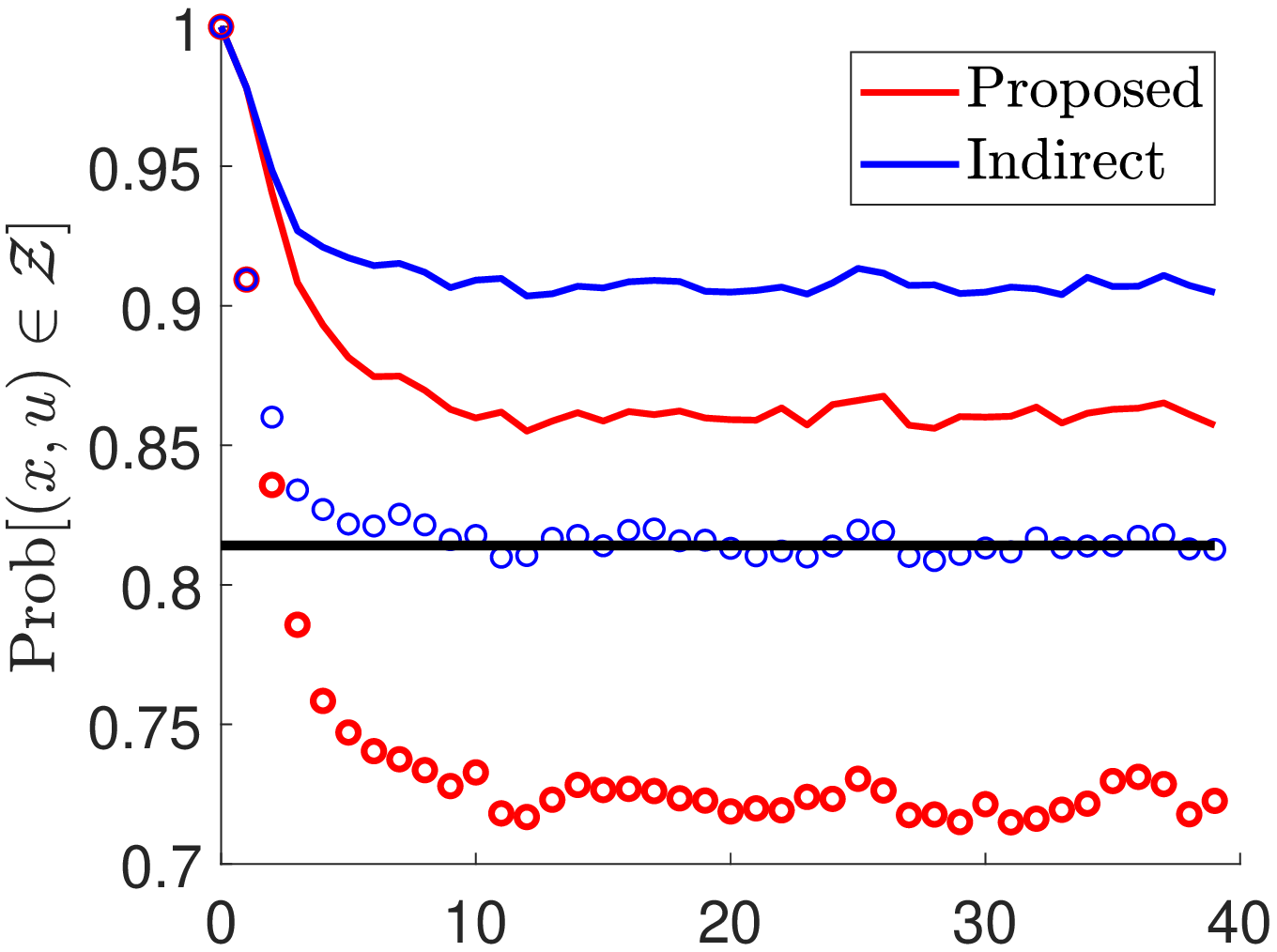}
\end{center}
\caption{Closed loop input $u$ and empirical probability of constraint violation $\prob{(x(k),u(k))\in\mathcal{Z}}$ over time $k\in\mathbb{I}_{[0,40]}$ for proposed (red) and indirect feedback (blue) SMPC. 
Mean trajectories are shown as solid lines and $40$ exemplary random realizations as dashed lines.
The mean for the SMPC implementations with  non-symmetric half-space PRS are shown in circles.} 
\label{fig:academic}
\end{figure}
Regarding the indirect feedback SMPC~\cite{hewing2020recursively}: we simply have $u\equiv 0$, and we cannot take advantage of beneficial disturbances.
If we consider the non-symmetric PRS in the indirect feedback SMPC, then we instead have $u\equiv -0.1624$ (except for $k=0$), which demonstrates the reduction in conservatism of using non-symmetric PRS.
For the proposed SMPC with symmetric PRS, the closed-loop input $u$ depends on the realized disturbances $w$. 
In fact, whenever a "positive" disturbance appears (both in the mathematical sense and in the practical sense that they are beneficial), then we re-set the initial state.
Hence, we can implicitly use the fact that at least $50\%$ of the disturbances are beneficial, which is also used in the non-symmetric PRS. 
Empirically, we achieve (on average) a similar performance (cf. Fig.~\ref{fig:academic}), i.e., the additional degree of freedom in the initial state results in a similar improvement as a non-symmetric PRS. 
Although the mean trajectories and hence expected performance is almost equivalent, the closed-loop realizations and statistics differ significantly. 
For example, the indirect-feedback SMPC exactly matches the prescribed chance constraint specification, while the proposed approach achieves a higher level of constraint satisfaction. 
Given the benefits of non-symmetric PRS, one might be tempted to also implement the proposed SMPC using non-symmetric PRS. 
However, in this case the proposed approach and the case distinction approach in~\cite{hewing2018stochastic} empirically fail to meet the chance constraint specification~\eqref{eq:chance_constraints} (cf. Fig.~\ref{fig:academic}) since the symmetry condition (Ass.~\ref{ass:PRS}) is violated.     
\end{document}